\documentclass[conference]{IEEEtran}
\IEEEoverridecommandlockouts
\def\BibTeX{{\rm B\kern-.05em{\sc i\kern-.025em b}\kern-.08em
    T\kern-.1667em\lower.7ex\hbox{E}\kern-.125emX}}

\usepackage[utf8]{inputenc} 
\usepackage[table,dvipsnames]{xcolor}
\usepackage{graphicx}
\graphicspath{{figs/}}

\usepackage{epstopdf}

\usepackage{tikz}
\usetikzlibrary{arrows}
\usepackage{tabularx}
\usepackage{dcolumn}
\usepackage{multirow}
\usepackage{pgfplots}
\usepackage{float}
\usepackage{amsthm}
\usepackage[cmex10]{amsmath}
\usepackage{amssymb}
\usepackage[numbers,sort&compress]{natbib}
\usepackage{braket}
\usepackage{array}
\usepackage{longtable}
\usepackage{arydshln}
\usepackage{bm}
\usepackage{bbm}

\DeclareMathOperator{\mmse}{mmse}

\usepackage[T1]{fontenc}
\usepackage{url}
\usepackage{ifthen}
\usepackage{cuted}

\makeatletter
\newcommand{\distas}[1]{\mathbin{\overset{#1}{\kern\z@\sim}}}%
\newsavebox{\mybox}\newsavebox{\mysim}
\newcommand{\distras}[1]{%
  \savebox{\mybox}{\hbox{\kern3pt$\scriptstyle#1$\kern3pt}}%
  \savebox{\mysim}{\hbox{$\sim$}}%
  \mathbin{\overset{#1}{\kern\z@\resizebox{\wd\mybox}{\ht\mysim}{$\sim$}}}%
}
\makeatother
\newtheorem{theorem}{Theorem}[section]
\newtheorem{lemma}[theorem]{Lemma}

\newtheorem{define}{Definiton}

\newtheorem{remark}{Remark}

\def\cF{\mathcal{F}}

\newcommand{\Pb}[1]{\mathbb{P}\left[#1 \right]}

\newcommand{\Ex}[1]{\mathbb{E}\left[#1 \right]}

\def\cn{\mathcal{CN}}

\def\cd{\mathcal{C}}

\def\Ieee{IEEEeqnarray*}
\def\Ieeen{\IEEEyesnumber}

\begin{document}

\title{Improved bounds for the many-user MAC\\
\thanks{This work was supported in part by the National Science Foundation under Grant CCF-17-17842.}
}

\author{\IEEEauthorblockN{Suhas S Kowshik}
\IEEEauthorblockA{\textit{Dept. of EECS} \\
\textit{Massachusetts Institute of Technology}\\
Cambridge, MA, USA \\
suhask@mit.edu}
}

\maketitle

\begin{abstract}
Many-user MAC is an important model for understanding energy efficiency of massive random access in 5G and beyond. Introduced in Polyanskiy'2017 for the AWGN channel, subsequent works have provided improved bounds on the asymptotic minimum energy-per-bit required to achieve a target per-user error at a given user density and payload, going beyond the AWGN setting. The best known rigorous bounds use spatially coupled codes along with the optimal AMP algorithm. But these bounds are infeasible to compute beyond a few (around 10) bits of payload. In this paper, we provide new achievability bounds for the many-user AWGN and quasi-static Rayleigh fading MACs using the spatially coupled codebook design along with a \emph{scalar} AMP algorithm. The obtained bounds are computable even up to 100 bits and outperform the previous ones at this payload. 
\end{abstract}


\section{Introduction}
Massive machine type communication (mMTC) is an important paradigm in 5G and beyond \cite{chen2020massive,mahmood2020white} where a large number of transmitters with small payloads communicate sporadically with the base station. This problem of massive multiple access was given an information theoretic footing in \cite{polyanskiy2017perspective}. Compared to the classical information theory of multiple access channels (MAC), the new formulation is distinguished in at least three aspects: 1) massive number of users compared to blocklength, 2) per-user probability of error (PUPE) metric and 3) random access. 

In particular, \cite{polyanskiy2017perspective} provided finite blocklength (FBL) bounds on minimum energy-per-bit $E_b/N_0$ required to achieve a target per-user error for the random access additive white Gaussian noise (AWGN) channel. This is also called the \emph{unsourced MAC}. Further, it also considered the \emph{many-user} asymptotics (also called the many-user MAC) where the number of users grow linearly with blocklength for fixed payload per user (without random access). The fundamental limit of minimum $E_b/N_0$ required to achieve a target PUPE on the many-user AWGN MAC, as a function of \emph{user density}, was demonstrated in \cite{polyanskiy2017perspective,ZPT-isit19} to undergo an interesting phase transition: for small values of user density, the minimum $E_b/N_0$ is \emph{almost constant} i.e., nearly the same as if there was a single user in the system. This corresponds to an almost perfect multi-user interference (MUI) cancellation. More importantly, the asymptotic performance turns out to be a reasonable proxy for the FBL behavior in the random-access setting as observed in \cite{polyanskiy2017perspective}. Later, the almost perfect MUI cancellation was observed for the many-user quasi-static Rayleigh fading (QSF) channel in \cite{kowshik2021fundamental,kowshik2020energy}. 

Recently there is a series of works that aim at constructing explicit practical coding schemes to achieve this MUI cancellation: see \cite{ordentlich2017low,amalladinne2020coded,kowshik2020energy,amalladinne2021enhanced,fengler2021sparcs,fengler2021non,fengler2021pilot,pradhan2020polar,pradhan2021ldpc} and references therein. But here we focus on the many-user MAC (i.e., the asymptotics) with the goal of obtaining improved and computable rigorous achievability bounds. To this end, \cite{hsieh2021near} used the approximate message passing (AMP) algorithm along with spatially coupled coding matrices to obtain improved achievability bounds for the many-user AWGN MAC. But those bounds are infeasible to evaluate for payload sizes exceeding 10 bits. Further, \cite{muller2021soft} considered practical schemes for many-user AWGN MAC based on interference cancellation and also provided asymptotic bounds, albeit in a non rigorous way.  

\paragraph*{Main contribution}We build on the spatial coupling idea from \cite{hsieh2021near} and the scalar AMP from \cite{kowshik2021fundamental} to provide new achievability bounds for the many-user AWGN and QSF MACs that are computable up to at least 100 bits of payload, which is now a standard in the unsourced MAC \cite{polyanskiy2017perspective, amalladinne2020coded,pradhan2021ldpc}. 

The structure of the paper is as follows. We define the system model in \ref{sec:sys_model}. The AMP algorithm is reviewed in \ref{sec:amp}, main results are provided in \ref{sec:main_results} and numerical computations are presented in~\ref{sec:numerical}.

\paragraph*{Notation} We denote by $\mathbb{N}$, $\mathbb{R}$ and $\mathbb{C}$ the sets of natural, real and complex numbers, respectively. For $n\in\mathbb{N}$ we let $[n]=\{1,2,\cdots n\}$. Euclidean norm is denoted as $\|\cdot\|$. For a matrix $A$, we use $A^{\top}$ and $A^*$ to denote the transpose and the Hermitian conjugate, respectively. Standard normal and circularly symmetric complex normal distributions are denoted by $\mathcal{N}(0,1)$ and $\mathcal{CN}(0,1)$, respectively. For $p\in [0,1]$, $\mathrm{BER}(p)$ denotes Bernoulli distribution with parameter $p$. $\mathrm{BG}(\sigma^2,p)$ denotes the (complex) Bernoulli-Gaussian distribution i.e., it is zero with probability $1-p$ and it is $\mathcal{CN}(0,\sigma^2)$ otherwise. Lastly, for a vector $X\in \mathbb{R}^n$ we let $X_{i:j}$ (with $i\leq j$) denote the sub-vector $(X_i,X_{i+1},\cdots,X_j)^\top$.

\section{System model}
\label{sec:sys_model}
Let $K$ denote the number of users and $n$ denote the blocklength. Let $\{P_{Y^n|X^n_1,\cdots,X_K^n}:\prod_{u=1}^K\mathcal{X}_i^n\to \mathcal{Y}^n\}_{n=1}^{\infty}$ denote a $K$-user MAC. Let $k$ be the payload size of each user (in bits) and $M=2^k$ is the total number of messages of each user. We let $W_u\in[M]$ denote the message of user $u$. The encoders and decoder are denoted by $f_u:[M]\to\mathcal{X}^n_u $ and $g:\mathcal{Y}^n\to [M]^K$, respectively. Let $X_u^n=f_u(W_u)$ denote the codeword transmitted by user $u$ and  $(\hat W_1,\cdots,\hat W_K)=g(Y^n)$ denote the decoded messages. In this work we consider the following two multiple access channels. 
\allowdisplaybreaks{
\begin{enumerate}
    \item \emph{AWGN MAC}: The channel $P_{Y^n|X_1^n,\cdots X_K^n}$ is given by
    \begin{equation}
        \label{eq:awgn_sys_model}
        Y^n=\sum_{u=1}^K X_u^n+Z^n
    \end{equation}
    where $X_u^n\in\mathbb{R}^n$ and $Z^n\distas{}\mathcal{N}(0,I_n)$ is the additive noise. 
    \item \emph{QSF MAC}: The channel $P_{Y^n|X_1^n,\cdots X_K^n}$ is given by
    \begin{equation}
        \label{eq:fading_sys_model}
        Y^n=\sum_{u=1}^K H_u X_u^n+Z^n
    \end{equation}
    where $X_u^n\in\mathbb{C}^n$ and $Z^n\distas{}\mathcal{CN}(0,I_n)$. Further $H_u\distas{i.i.d.}\cn(0,1)$ are the fading coefficients. We assume that both the transmitters and the receiver know the distribution of the fading coefficients, but the realizations are \emph{unknown} -- this is a no-CSI fading channel.
\end{enumerate}}
For both channels, we impose a natural \emph{power constraint} $\|X_u^n\|^2\leq nP,\, \forall u\in[K]$.
The error metric is the \emph{per-user probability of error} (PUPE) 
\begin{equation}
    \label{eq:pupe_def}
    P_e=\frac{1}{K}\sum_{u=1}^K\Pb{W_u\neq \hat W_u}.
\end{equation}

\begin{remark}
We use the subscript $u$ to index users and $i$, $j$ etc. to denote a particular entry in a vector unless the distinction is unclear. Also we will suppress the superscript $n$ for brevity. 
\end{remark}

\subsection{Formulation as a compressed sensing problem}
It is well known that the MAC can also be modelled as a compressed sensing problem \cite{aeron2010information,jin2011limits,ZPT-isit19, kowshik2021fundamental,hsieh2021near}. We describe this connection here since it forms the basis of our results. 

Let $p=KM$ and $A=[A_1,\cdots,A_K]$ be the $n\times p $ matrix formed by the concatenation of the codebooks of all users. Here columns of $A_u$ denote the codewords of user $u$. Let $U$ denote a length $p$ vector. Let $S\in\{0,1\}^p$ denote the support vector of $U$ i.e., $S_i=1[U_i\neq 0]$. $U$ is \emph{block-sparse}: for each $u\in [K]$, $\sum_{j=(u-1)M+1}^{u M}S_j\leq 1$. The system models for the two channels can be equivalently described by 
\begin{equation}
        \label{eq:awgn_cs_model}
        Y=AU+Z
    \end{equation}
    where $Z$ is the noise described after \eqref{eq:awgn_sys_model} and \eqref{eq:fading_sys_model} for the AWGN and QSF MAC, respectively. The specifications of the nonzero entry in each block or section of $U$ (of size $M$) are:
\begin{enumerate}
    \item \emph{AWGN MAC}: $U\in\{0,1\}^p$ with $\sum_{j=(i-1)M+1}^{i M}U_j=1$. 
     
    \item \emph{QSF MAC}: $U\in \mathbb{C}^p$ with the nonzero entry in section $u\in[K]$ of $U$ to be set to $H_u$ -- the fading coefficient of user $u$. Thus $\sum_{j=(u-1)M+1}^{u M}S_j= 1 \, a.s.$ 
\end{enumerate}

Thus the MAC decoding reduces to support recovery in \eqref{eq:awgn_cs_model}. Letting $\hat S$ denote estimated support, the PUPE is the expected \emph{section error rate} SER
\begin{equation}
    \label{eq:ser_defn}
    SER=\frac{1}{K}\sum_{u=1}^K 1[S_{(u-1)M+1:uM}\neq \hat S_{(u-1)M+1:uM}]
\end{equation}

\subsection{Many-user limit}
As alluded to in the introduction, we aim to understand the fundamental limit when $K=\mu n$ and $n\to\infty$ with \emph{fixed energy-per-bit} $E_b/N_0$ denoted by $\mathcal{E}$. Recall that $k=\log_2 M$ is the payload size. The $E_b/N_0$ is defined as:
\begin{enumerate}
    \item \emph{AWGN MAC}: $\mathcal{E}=\frac{nP}{2k}$. The corresponding power constraint is $\|X_u\|^2\leq 2\mathcal{E}k$.
    \item \emph{QSF MAC}: $\mathcal{E}=\frac{nP}{k}$. The corresponding power constraint is $\|X_u\|^2\leq \mathcal{E}k$.
\end{enumerate}
We let $E=nP$ be the total energy. That is $E=2\mathcal{E}k$ for the AWGN MAC and $E=\mathcal{E}k$ for the quasi-static fading MAC.

\begin{define}
An $(n,M,\epsilon,\mathcal{E},K)$ code for the $K$-user MAC \eqref{eq:awgn_sys_model} or \eqref{eq:fading_sys_model} is a collection of codebooks $\{\cd_u:u\in[K]\}$ of size $M$ each, along with a decoder such that the codewords satisfy the power constraints (set by $\mathcal{E}$) and the PUPE is smaller than $\epsilon$.
\end{define}
This leads to the following fundamental limit.
\begin{define}
The fundamental limit of $E_b/N_0$ for a given payload $k$, target PUPE $\epsilon$ and user density $\mu$ is defined as 
\begin{equation}
    \mathcal{E}^*=\limsup_{n\to\infty}\inf\{\mathcal{E}:\exists (n,M,\epsilon,\mathcal{E},K=\mu n)-\mathrm{code} \}
\end{equation}
\end{define}

Next we describe AMP and provide new upper bounds on $\mathcal{E}^*$.

\section{Approximate message passing (AMP)}
\label{sec:amp}
AMP are a class of low complexity iterative algorithms introduced in \cite{donoho2009message} for signal recover in compressed sensing, or more generally for statistical inference on models based on dense factor graphs. The behavior of AMP in the high dimensional limit is tracked by the \emph{state evolution} equations. The convergence of AMP parameters to the state evolution has been proved under various assumptions (see \cite{bayati2011dynamics,javanmard2013state,ma2019analysis,berthier2020state,rush2018finite,rush2021capacity,gerbelot2021graph}). Furthermore AMP has been successful as a near optimal decoder for sparse superposition codes \cite{joseph2012least,barbier2015statistical,barbier2014replica,barbier2017approximate,rush2021capacity}.

In the context of massive multiple access, AMP algorithms have found applications in unsourced MAC \cite{fengler2021sparcs,amalladinne2020unsourced}, but also to provide achievability bounds in the many-user asymptotics \cite{kowshik2021fundamental, hsieh2021near}. In this work we build upon \cite{kowshik2021fundamental, hsieh2021near} to provide new, improved and computable achievability bounds for the two channels considered. In particular, \cite{hsieh2021near} in essence provided bounds on the Bayes optimal decoder (for the i.i.d. Gaussian coding matrix, see \cite[Remark 3.3]{hsieh2021near}) for the many-user AWGN channel using the spatially coupled coding matrix $A$ along with the AMP algorithm that uses the the optimal section wise denoiser. The same bounds are also directly obtained (but non-rigorously) from \cite{barbier2017approximate}. But these bounds can only be computed for small values of $k$ since it involves evaluating $M=2^k$ dimensional integrals. On the other hand, \cite{kowshik2021fundamental} considered a \emph{scalar AMP} algorithm that ignores the block sparse structure of $U$ in \eqref{eq:awgn_cs_model} along with i.i.d Gaussian coding matrix $A$. The bounds obtained this way are near optimal only for small values of $\mu$. 

In this paper we use scalar AMP idea from \cite{kowshik2021fundamental} and spatial coupling from \cite{hsieh2021near} to obtain improved achievability bounds that are computable for moderate values of $k$ (like $k=100$ bits that is a standard in massive and unsourced MAC \cite{polyanskiy2017perspective,amalladinne2020coded,amalladinne2020unsourced}).

\subsection{Spatially coupled codebook}
\label{sec:spatial_coupl_code}
Now we describe the spatially coupled codebook design based on \cite{hsieh2021near}. Let $R,C\in \mathbb{N}$ be such that $R$ divides $n$ and $C$ divides $p=KM$. The codebook $A$ is divided into blocks of size $\frac{n}{R}\times \frac{p}{C}$ and hence can be considered as a block matrix of size $R\times C$. Let $B\in \mathbb{R}^{R\times C }$ be the base matrix with nonnegative entries $B_{r,c}$ such that $\sum_{r=1}^R B_{r,c}=1$ for all $c\in[C]$. Further, with abuse of notation, let $r:[n]\to [R]$ and $c:p\to [C]$ denote functions that map a particular row or column index to its corresponding block. Then the matrix $A$ is constructed as 
\begin{enumerate}
    \item \emph{AWGN MAC}: $A_{i,j}\distas{i.i.d.}\mathcal{N}\left(0,E(R/n)B_{r(i),c(j)}\right)$
    \item \emph{QSF MAC}: $A_{i,j}\distas{i.i.d.}\mathcal{CN}\left(0,E(R/n)B_{r(i),c(j)}\right)$
\end{enumerate}
 In particular we use the $(\omega,\Lambda,\rho)$ base matrix from \cite{hsieh2021near} as the choice of $B$ which is as follows. Let $\rho\in [0,1)$, $\omega\geq 1$ and $\Lambda\geq 2\omega-1$. Then we choose $R=\Lambda+\omega-1$ and $C=\Lambda$. Finally we have 
\begin{equation}
\label{eq:coupling_base_matrix}
B_{r,c}=\begin{cases}
\frac{1-\rho}{\omega}, & c\leq r \leq c+\omega -1\\
\frac{\rho}{\Lambda-1}, & \mathrm{o/w}
\end{cases}
\end{equation}

Let $\tilde \mu=\frac{R}{C}\mu$ be the effective user density. Since usually $\omega> 1$ we have that $\tilde\mu>\mu$. 

\subsection{Scalar equivalent channel}
We define the equivalent scalar channel necessary to describe the AMP and the state evolution. It is a scalar AWGN channel parameterized in terms of noise variance $\sigma^2$:
\begin{equation}
    \label{eq:scalar_model}
    V_{\sigma^2}=X+\sigma W
\end{equation}
where $X$ independent of $W$, and 
\begin{enumerate}
    \item \emph{AWGN MAC}: $X\distas{}\mathrm{BER}(1/M)$, $W\distas{}\mathcal{N}(0,1)$.
    \item \emph{QSF MAC}: $X\distas{}\mathrm{BG}(1,1/M)$, $W\distas{}\mathcal{CN}(0,1)$.
\end{enumerate}
We denote the joint distribution of $X$ and $V_{\sigma^2}$ by $P_{X,V_{\sigma^2}}$. For each of the above scalar channels, the corresponding denoising function is
\begin{equation}
    \label{eq:denoise_defn}
    \eta(v,\sigma^2)=\Ex{X|V_{\sigma^2}=v}=\Ex{X|X+\sigma W=v}.
\end{equation}
The minimum mean squared error of estimating $X$ from $V$ is given by 
\begin{\Ieee}{LLL}
    \label{eq:mmse_defn}
    \mmse(\sigma^2)&=&\Ex{(X-\eta(V_{\sigma^2},\sigma^2))^2}\Ieeen
\end{\Ieee}

Lastly we define the equivalent of support recovery. Let $S_0=1[X\neq 0]$. Let $\hat S_0(\theta)$ be an estimate of $S_0$ based on observation $V_{\sigma^2}$. In particular we use the following estimators
\begin{enumerate}
    \item \emph{AWGN MAC}: $\hat S_0(\theta)=1[V_{\sigma^2}>\theta]$
    \item \emph{QSF MAC}: $\hat S_0(\theta)=1[|V_{\sigma^2}|^2>\theta]$
\end{enumerate}
Then we denote the probability of error in support recovery by $\psi$:
\begin{equation}
    \label{eq:amp_scalar_pe}
    \psi(\sigma^2,\theta,M)=\Pb{S_0\neq \hat S_0(\theta)}
\end{equation}

\subsection{Algorithm}
We describe the AMP algorithm for both AWGN MAC and QSF MAC. The variables appearing in the descriptions must be interpreted accordingly. The version of AMP described here is adapted from \cite{hsieh2021near}.

Start with $U^{(0)}=0\in\mathbb{C}^{p}$, $R^{(0)}=Y$. Then for $t=1,2,\cdots $ we have the following iterations
\begin{\Ieee}{LLL}
U^{(t)}= \eta^{(t)}\left((\tilde  Q^{(t-1)} \odot A)^* R^{(t-1)}+U^{(t-1)}\right)\label{eq:sc_amp_iter1}\Ieeen\\
R^{(t)} = Y-AU^{(t)}+\frac{R}{C}\mu M (\tilde b^{(t)} \odot R^{(t-1)})\label{eq:sc_amp_iter2}\Ieeen
\end{\Ieee}
where $\odot $ denotes element wise product, and matrix $\tilde  Q^{(t)}$, vector $\tilde b^{(t)}$ and denoiser $\eta^{(t)}$ will be defined next via the state evolution. 

Let $\psi_c^{(0)}=\infty$. Then for $t\geq 1$, for each $r\in[R]$ and $c\in[C]$ we define
\begin{\Ieee}{RLLLLLLL}
\gamma_r^{(t)} &= &\sum_{c=1}^C B_{r,c}\psi_c^{(t)}, &\quad & \phi_r^{(t)} &=& \frac{1}{E}+\tilde\mu M\gamma_r^{(t)}\label{eq:sc_amp_se1}\Ieeen\\
\tau_c^{(t)} &=&\frac{1}{\sum_{r=1}^R B_{r,c}\left(\phi_r^{(t)}\right)^{-1}}\label{eq:sc_amp_se3}, &\quad & \psi_c^{(t+1)}&=&\mmse(\tau_c^{(t)})\label{eq:sc_amp_se4}\Ieeen
\end{\Ieee}
where $\mmse(\cdot)$ is defined in \eqref{eq:mmse_defn}. Now the matrices $\tilde Q^{(t)}$ and vectors $\tilde b^{(t)}$ are defined as follows. For each $i\in [n]$ and $j\in [KM]$
\begin{\Ieee}{LLL}
\tilde b^{(t)}_{i}=\tilde\mu M \frac{\gamma^{(t)}_{r(i)}}{\phi^{(t-1)}_{r(i)}} \quad \tilde Q^{(t)}_{i,j}=\frac{\tau^{(t)}_{c(j)}}{\phi^{(t)}_{r(i)}}\Ieeen
\end{\Ieee}
The denoiser at time $t$ is given by $\eta^{(t)}=(\eta^{(t)}_1,\cdots,\eta^{(t)}_{p})$ with $\eta^{(t)}_i(z)=\eta(z,\tau^{(t)}_{c(i)})$ and $\eta$ defined in \eqref{eq:denoise_defn}. The estimate of $U$ after $t$ steps is 
given by (see \cite{hsieh2021near} for details on hard decision estimate) 
\begin{equation}
\label{eq:sc_amp_iter4}
\hat U^{(t)}=(\tilde Q^t \odot A)^* R^{(t)}+U^{(t)}
\end{equation}
To convert $\hat U^{(t)}$ into support $\hat S^{(t)}$ we perform a simple thresholding for each $c\in [C]$ i.e., for each $i$
\begin{enumerate}
    \item \emph{AWGN MAC}: 
    \begin{equation}\label{eq:sc_amp_decoder_awgn}
	\hat S^{(t)}_i(\theta_{c(i)}) = 1[\hat U_i^{(t)}>\theta_{c(i)}]
\end{equation}
\item \emph{QSF MAC}:
\begin{equation}\label{eq:sc_amp_decoder_fading}
	\hat S^{(t)}_i(\theta_{c(i)}) = 1[|\hat U_i^{(t)}|^2>\theta_{c(i)}]
\end{equation}
\end{enumerate}

where $\{\theta_c:c\in [C]\}$ is a set of thresholds. 

\section{Main results}
\label{sec:main_results}
First we state a lemma that follows directly from \cite[Theorem 2]{hsieh2021near} (with $B=1$ in their notation which in turn relies on \cite{yedla2014simple}). Define the replica potential 
\begin{\Ieee}{LLL}
\label{eq:sc_amp_replica_pot_fading}
\cF_{\mathsf{QSF}}(\tau) &=&(\mu M) I(X;V_{\tau})+\left(\ln \tau +\frac{1}{\tau E} -1\right)\Ieeen
\end{\Ieee}
where $(X,V_{\tau})\distas{}P_{X,V_{\tau}}$. Further, let $\mathcal{M}$ denote the maximum of the global minimizers of $\cF$:
\begin{equation}
\label{eq:sc_amp_replica_argmin}
\mathcal{M}_{\mathsf{QSF}}(\mu,E,M)=\max(\arg\min_{\tau>\frac{1}{E}}\cF_{\mathsf{QSF}}(\tau))
\end{equation}

Similarly, we have the potential for the AWGN MAC
\begin{\Ieee}{LLL}
\label{eq:sc_amp_replica_pot_awgn}
\cF_{\mathsf{AWGN}}(\tau) &=&(\mu M) I(X;V_{\tau})+\frac{1}{2}\left(\ln \tau +\frac{1}{\tau E} -1\right)\Ieeen
\end{\Ieee}
where $(X,V_{\tau})\distas{}P_{X,V_{\tau}}$ (corresponding to AWGN). We let $\mathcal{M}_{\mathsf{AWGN}}$ denote the maximum of the global minimizers of $\cF_{\mathsf{AWGN}}$.

\begin{lemma}
\label{lem:sc_amp_lemma_1}
 For any $(\omega,\Lambda,\rho)$ base matrix $B$, for each $c\in[C]$, $\tau_c^{(t)}$ is non-increasing in $t$ and converges to a fixed point $\tau_c^{\infty}$. Furthermore, for any $\delta>0$, there exists $\omega_0<\infty$, $\Lambda_0<\infty$ and $\rho_0>0$ such that for all $\omega>\omega_0$, $\Lambda>\Lambda_0$ and $\rho<\rho_0$, the fixed points $\{\tau_c^{\infty}:c\in[C]\}$ satisfy
\begin{equation}
\label{eq:sc_amp_fp_ub}
\tau^{\infty}_c\leq \tau^{\infty}\left(\tilde \mu\right)+\tilde \mu M\delta
\end{equation}
where $\tau^{(\infty)}(\tilde \mu)=\mathcal{M}_{\mathsf{QSF}}\left(\tilde\mu,E,M\right)$ for the QSF MAC and $\tau^{(\infty)}(\tilde \mu)=\mathcal{M}_{\mathsf{AWGN}}\left(\tilde\mu,E,M\right)$ for the AWGN MAC.

\end{lemma}

$\tau_c^{(t)}$ tracks the noise variance (and hence also mmse) of estimation in the scalar channel \eqref{eq:scalar_model}. Thus the fixed points of the spatially coupled system are at least as good as the uncoupled system (i.e., with $A$ having i.i.d entries) but with user density increased from $\mu$ to $\tilde\mu$. If we take limits as $\Lambda\to\infty$ \emph{and then} $\omega\to \infty$ we obtain that $\tau^{\infty}(\tilde\mu)\to\tau^{\infty}(\mu)$. This is known as threshold saturation (see \cite[Remark 3.3]{hsieh2021near}). 

\subsection{QSF MAC}

We present the main achievability bound for the QSF MAC.
\begin{theorem}
\label{thm:sc_amp_ach}
Fix any $\mu>0$, $E>0$ and $k=\log_2 M\ge 1$. Then for every $\mathcal{E} >\frac{E}{k}$ there exist a sequence of $\left(n,M,\epsilon_n,\mathcal{E},K=\mu n\right)$ codes for the QSF MAC such that 
\begin{equation}
    \label{eq:thm_qsf}
    \limsup_{n\to\infty} \epsilon_n \le \pi^*(\tau^{(\infty)}(\mu), M)
\end{equation}
where $\pi^*(\tau, M) = 1-\frac{1}{1+\tau} \left((M-1)\left(\frac{1}{\tau}+1\right)\right)^{-\tau}$ 
and 
 \begin{\Ieee}{LLL}
 \label{eq:sc_amp_fp_1}
 \tau^{(\infty)}(\mu) &\equiv & \tau^{(\infty)}(\mu;E,M)=\mathcal{M}_{\mathsf{QSF}}(\mu,E,M)\Ieeen
 \end{\Ieee}
\end{theorem}

\begin{proof}
The idea is to use random coding along with the spatially coupled codebook described in \ref{sec:spatial_coupl_code}. The proof is similar to that of \cite[Theorem IV.6]{kowshik2021fundamental} but uses the result on convergence of the empirical joint distribution of entries in $(U,\hat U^{(t)})$ in the spatially coupled systems from \cite{donoho2013information,javanmard2013state} (adapted to the complex number setting). For $S,\hat S\in \{0,1\}^p$ the hamming distance is given by $d_{H}(S,\hat S)=\frac{1}{p}\sum_{i=1}^{p}1[S_i\neq \hat S_i]$. Recall that if $S$ is the support of the true signal $U$, and $\hat S^{(t)}\equiv (\hat S^{(t)}_i(\theta_{c(i)}))_{i=1}^{p}$ (see \eqref{eq:sc_amp_decoder_fading}) is the estimate of the support, then from \cite[eqn. (119)]{kowshik2021fundamental} we have that 
\begin{equation}
\label{eqL:sc_amp_pupe_ham}
\text{PUPE}(\hat S^{(t)})\leq M\Ex{d_H(S,\hat S^{(t)}))}
\end{equation}

Notice that 
$$d_H(S,\hat S^{(t)})=\frac{1}{C}\sum_{c=1}^C\biggr[ \frac{C}{p}\sum_{i=(c-1)\frac{p}{C}+1}^{c\frac{p}{C}}1[S_i\neq \hat S^{(t)}_i]\biggr]$$

Moreover, from \cite[Theorem 1]{javanmard2013state} (see proof of lemma 1 there) we have that for any Lipschitz function $f:\mathbb{C}^2\to \mathbb{R}$ (or more generally any pseudo-Lipschitz function \cite{bayati2011dynamics}) the following holds almost surely (with $K=\mu n$):
\begin{equation}
\label{eq:sc_amp_se_cgt}
\lim_{n\to\infty}\frac{C}{p}\sum_{i=(c-1)\frac{p}{C}+1}^{c\frac{p}{C}}f(U_i,\hat U^{(t)}_i)=\Ex{f(X,V_{\tau^{(t)}_c})}
\end{equation}
where $(X,V_{\tau^{(t)}_c})~\distas{}P_{X,V_{\tau^{(t)}_c}}$.

\begin{remark}
Although \cite{javanmard2013state} deal only with real valued system, as noted in \cite[Sec 4.4]{hsieh2021spatially}, the proofs in \cite{javanmard2013state} go through for complex valued systems as well. 
\end{remark}
Standard approximation argument \cite[Theorem 1(3)]{hsieh2021near} gives 
\begin{equation}
\label{eq:sc_amp_ham_cgt}
\lim_{n\to\infty}\frac{C}{p}\sum_{i=(c-1)\frac{p}{C}+1}^{c\frac{p}{C}}\Pb{S_i\neq \hat S^{(t)}_i}=\psi(\tau^{(t)}_c,\theta_c,M)
\end{equation}
where the $\psi()$ is defined in \eqref{eq:amp_scalar_pe} and $\{\theta_c:c\in[C]\}$ are thresholds \eqref{eq:sc_amp_decoder_fading}. Thus for any $\{\theta_c>0:c\in[C]\}$
\begin{equation}
\label{eq:sc_amp_pupe_lim_ub}
\lim_{n\to\infty}\text{PUPE}(\hat S^{(t)})\leq \frac{1}{C}\sum_{c=1}^C M\psi(\tau^{(t)}_c,\theta_c,M)
\end{equation}
Now we take $t\to\infty$ and use lemma~\ref{lem:sc_amp_lemma_1} to obtain
\begin{equation}
\label{eq:sc_amp_pupe_lim_ub_1}
\lim_{t\to\infty}\lim_{n\to\infty}\text{PUPE}(\hat S^{(t)})\leq \frac{1}{C}\sum_{c=1}^C M\psi(\tau^{(\infty)}_c,\theta_c,M)
\end{equation}
Since $\{\theta_c\}$ are arbitrary we can minimize over $\{\theta_c>0:c\in[C]\}$ and use~\cite[Claim 6]{kowshik2021fundamental} to obtain
\begin{equation}
\label{eq:sc_amp_pupe_lim_ub_opt_1}
\lim_{t\to\infty}\lim_{n\to\infty}\text{PUPE}(\hat S^{(t)})\leq \frac{1}{C}\sum_{c=1}^C \pi^*(\tau^{(\infty)}_c,M)
\end{equation}
where $\pi^*(\tau,M)$ is described in the statement of the theorem.

Since $\pi^*$ is non-decreasing in $\tau$, from the second item in  lemma~\ref{lem:sc_amp_lemma_1} we have that for any fixed $\delta>0$, for all large enough $\omega$ and $\Lambda$, and all small enough $\rho$:
\begin{equation}
\label{eq:sc_amp_pupe_lim_ub_opt_2}
\lim_{t\to\infty}\lim_{n\to\infty}\text{PUPE}(\hat S^{(t)})\leq \pi^*(\tau^{(\infty)}(\tilde \mu)+\tilde \mu M\delta,M)
\end{equation}

Taking limit as $\Lambda\to\infty$ and then $\omega\to\infty$ we obtain that for every $\delta>0$ there is a $\rho_0>0$ such that for all $0<\rho<\rho_0$,
\begin{equation}
\label{eq:sc_amp_pupe_lim_ub_opt_3}
\lim_{\omega\to\infty}\lim_{\Lambda\to\infty}\lim_{t\to\infty}\lim_{n\to\infty}\text{PUPE}(\hat S^{(t)})\leq \pi^*(\tau^{(\infty)}(\mu)+\mu M\delta,M)
\end{equation}

The theorem is proved by noticing that $\delta>0$ is arbitrary.

\end{proof}

\subsection{AWGN MAC}
Next we have the achievability bound for the AWGN MAC.


\begin{theorem}
\label{thm:sc_amp_ach_awgn}
Fix any $\mu>0$, $E>0$ and $k=\log_2 M\ge 1$. Then for every $\mathcal{E} >\frac{E}{2k}$ there exist a sequence of $\left(n,M,\epsilon_n,\mathcal{E},K=\mu n\right)$ codes for the AWGN MAC 
$$ \limsup_{n\to\infty} \epsilon_n \le 2\epsilon^*(\tau^{(\infty)}(\mu), M)\,,$$
where $\epsilon^*(\tau,M)$ is the solution to 
\begin{equation}
\label{eq:sc_amp_pupe_fp}
\frac{1}{\sqrt{\tau}}=\mathcal{Q}^{-1}\left(\epsilon^*\right)+\mathcal{Q}^{-1}\left(\frac{\epsilon^*}{M-1}\right)
\end{equation}
and $\tau^{(\infty)}(\mu)=\mathcal{M}_{\mathsf{AWGN}}(\mu,E,M)$.
\end{theorem}

\begin{proof}
The proof follows from random coding using spatially coupled codebooks from section~\ref{sec:spatial_coupl_code} and is similar to that of theorem~\ref{thm:sc_amp_ach} and hence we will only highlight key steps. In particular, we have 

\begin{equation}
\label{eqL:sc_amp_awgn_pupe_ham}
\text{PUPE}(\hat S^{(t)})\leq M\Ex{d_H(U,\hat S^{(t)}))}
\end{equation}

From state evolution, it can be shown that
\begin{equation}
\label{eq:sc_amp_awgn_ham_cgt}
\lim_{n\to\infty}\frac{C}{p}\sum_{i=(c-1)\frac{p}{C}+1}^{c\frac{p}{C}}\Pb{U_i\neq \hat S^{(t)}_i}=\Pb{X\neq \hat S_0} 
\end{equation}
where $\hat S_0$ is from \eqref{eq:sc_amp_decoder_awgn}. Notice that the Bayes' optimal estimator for $X$ is of the form $\hat S_0$ for some carefully chosen $\theta_c$. 

As in the proof of theorem~\ref{thm:sc_amp_ach}, we take limit as $t\to\infty$, apply lemma~\ref{lem:sc_amp_lemma_1} and then optimize over $\theta_c$ we obtain
\begin{equation}
\label{eq:sc_amp_pupe_awgn_lim_ub_opt}
\lim_{t\to\infty}\lim_{n\to\infty}\text{PUPE}(\hat S^{(t)})\leq \frac{1}{C}\sum_{c=1}^C M\tilde \epsilon^*(\tau^{(\infty)}_c,M)
\end{equation}
where $\tilde\epsilon^*(\tau,M)$ is the minimum probability of error for decoding $X$ from the scalar channel \eqref{eq:scalar_model} corresponding to AWGN MAC. It can be shown that $\tilde\epsilon^*(\tau,M)$ satisfies
\begin{equation}
\label{eq:tilde_epsilon_star}
\frac{1}{\sqrt{\tau}}=\mathcal{Q}^{-1}\left(\frac{M\tilde \epsilon^*}{2}\right)+\mathcal{Q}^{-1}\left(\frac{M\tilde \epsilon^*}{2(M-1)}\right)
\end{equation}

Let $\epsilon^*(\tau, M)=\frac{M\tilde\epsilon^*(\tau,M)}{2}$. Using monotonicity of $\epsilon^*$ with respect to $\tau$ and lemma~\ref{lem:sc_amp_lemma_1} (along with threshold saturation) concludes the proof.
\end{proof}
\begin{figure}[hb]
  \begin{center}
     \includegraphics *[width=\columnwidth]{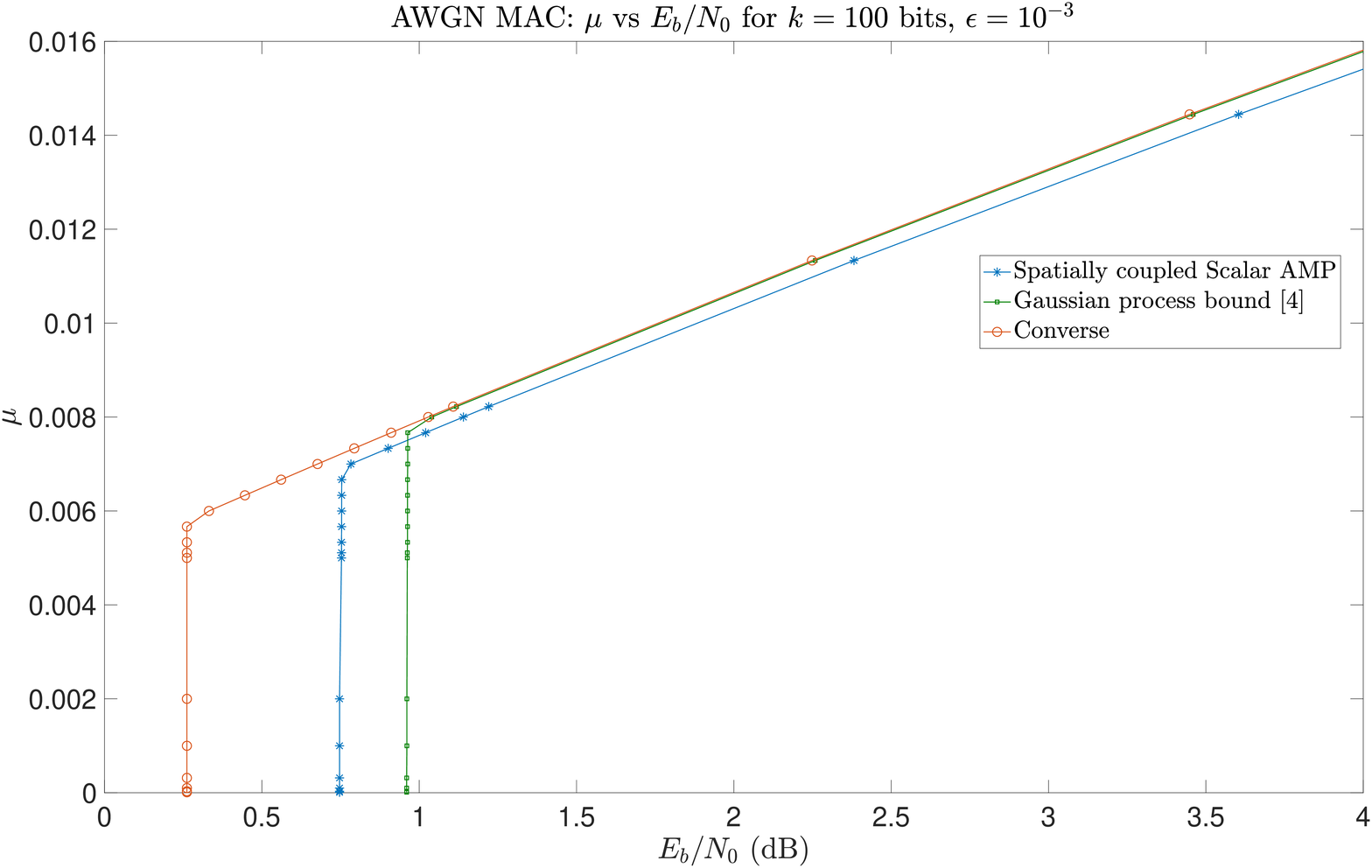}
     \caption {AWGN MAC: $\mu$ vs $E_b/N_0$ for $\epsilon\leq 10^{-3}$, $k=100$}
     \label{fig:1}
      \end{center} 
 \end{figure}

\begin{figure}[ht]
  \begin{center}
     \includegraphics *[width=\columnwidth]{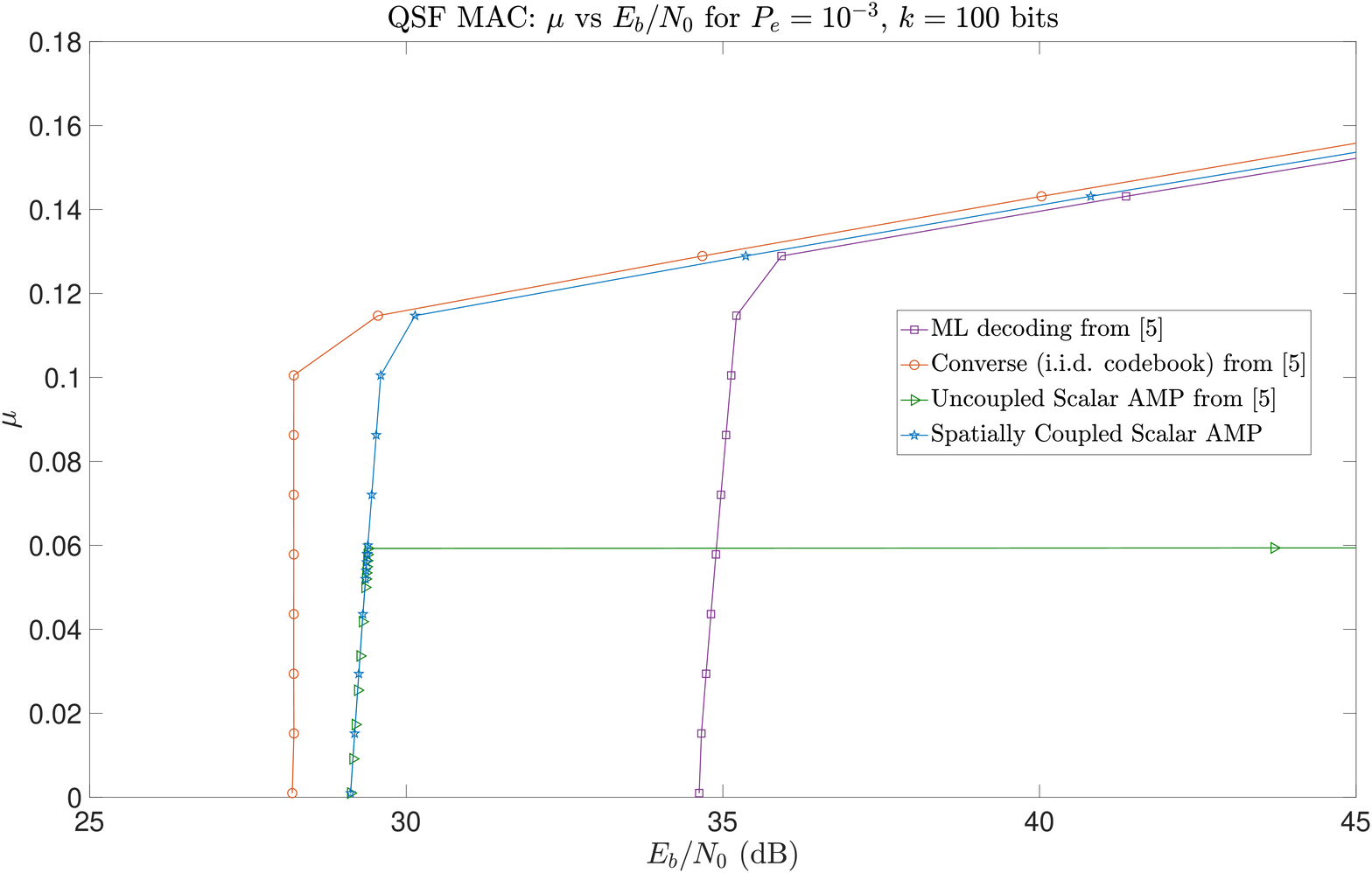}
     \caption {QSF MAC: $\mu$ vs $E_b/N_0$ for $\epsilon\leq 10^{-3}$, $k=100$}
     \label{fig:2}
      \end{center} 
 \end{figure}
\section{Numerical evaluation}
\label{sec:numerical}
Numerical evaluation of the bounds in theorems~\ref{thm:sc_amp_ach_awgn} and \ref{thm:sc_amp_ach} are shown in Fig.~\ref{fig:1} and Fig.~\ref{fig:2}, respectively. The parameters considered are similar to the previous works in many-user MAC \cite{polyanskiy2017perspective,ZPT-isit19,kowshik2021fundamental}: we set $k=100$ bits and target PUPE $\epsilon=10^{-3}$. Our bounds outperform the previous bounds on QSF MAC from \cite{kowshik2021fundamental}. For the AWGN MAC, our bounds are superior compared to \cite{ZPT-isit19} in the \emph{vertical} regime of the $\mu$ vs $E_b/N_0$ curves. We emphasize here that this vertical portion is the most relevant since in depicts almost perfect MUI cancellation.

\section*{Acknowledgment}
The author likes to thank Prof.~Yury Polyanskiy for numerous helpful and stimulating discussions. 

\bibliographystyle{IEEEtran}
	\bibliography{refs}
\end{document}